\documentclass[preprint,12pt,3p]{elsarticle}
\usepackage{amssymb}
\usepackage{amsthm}
\usepackage{graphicx}
\usepackage{amsmath}
\usepackage{hyperref}
\usepackage{graphicx}

\usepackage{pgf,tikz}

\usetikzlibrary{snakes}
\tikzstyle{printersafe}=[snake=snake,segment amplitude=0 pt]

\usetikzlibrary{arrows}
\usetikzlibrary{shapes}
\usepackage{multirow}
\usepackage{latexsym}
\usepackage{setspace}
\PassOptionsToPackage{boxed,section}{algorithm}
\usepackage{algorithm}
\usepackage{algorithmic}
\usepackage{enumerate}
\usepackage{amsmath}			% special AMS math definitions
\usepackage{amssymb}			% special AMS math symbols
\usepackage{url}
\usepackage{setspace}
\usepackage{tikz}
\usetikzlibrary{arrows}
\usetikzlibrary{shapes}
\usetikzlibrary{decorations.markings}
\usepackage{geometry}
\usepackage{bbm}

\newtheorem{proposition}{\em Proposition}

\newtheorem{theorem}{\em Theorem}
\newtheorem{conjecture}{\em Conjecture}
\newtheorem{definition}{\em Definition}

\newtheorem{remark}{\em Remark}
\newtheorem{corollary}{\em Corollary}
\newtheorem{observation}{\em Observation}

\journal{Sample Journal}

\begin{document}

\begin{frontmatter}

\title{On the fixed-parameter tractability of the partial vertex cover problem with a matching constraint in edge-weighted bipartite graphs}
%\tnotetext[label0]{This is only an example}

\author[label1]{Vahan Mkrtchyan\corref{cor1}}
\address[label1]{Gran Sasso Science Institute,
School of Advanced Studies, L'Aquila, Italy}
%\address[label2]{Address Two\fnref{label4}}

\cortext[cor1]{I am corresponding author}
%\fntext[label3]{I also want to inform about\ldots}
%\fntext[label4]{Small city}

\ead{vahan.mkrtchyan@gssi.it}
%\ead[url]{author-one-homepage.com}

\author[label5]{Garik Petrosyan}
\address[label5]{Department of Informatics and Applied Mathematics,
Yerevan State University, Yerevan, Armenia}
\ead{garik.petrosyan.1@gmail.com}

%\author[label1,label5]{Author Three}
%\ead{author.three@mail.com}

\begin{abstract}
In the classical partial vertex cover problem, we are given a graph $G$ and two positive integers $R$ and $L$. The goal is to check whether there is a subset $V'$ of $V$ of size at most $R$, such that $V'$ covers at least $L$ edges of $G$. The problem is NP-hard as it includes the Vertex Cover problem. Previous research has addressed the extension of this problem where one has weight-functions defined on sets of vertices and edges of $G$. In this paper, we consider the following version of the problem where on the input we are given an edge-weighted bipartite graph $G$, and three positive integers $R$, $S$ and $T$. The goal is to check whether $G$ has a subset $V'$ of vertices of $G$ of size at most $R$, such that the edges of $G$ covered by $V'$ have weight at least $S$ and they include a matching of weight at least $T$. In the paper, we address this problem from the perspective of fixed-parameter tractability. One of our hardness results is obtained via a reduction from the bi-objective knapsack problem, which we show to be W[1]-hard with respect to one of parameters. We believe that this problem might be useful in obtaining similar results in other situations. 
\end{abstract}


\end{frontmatter}

% \linenumbers


\section{Introduction}

In the present paper, we study an extension of the classical vertex cover problem (VC). We address this problem from
the perspective of parameterized complexity theory and algorithms. Our problem deals with the partial vertex cover problem (PVC). Recall that in this problem the goal is to cover a certain number of edges (not necessarily all of the edges of the input graph as in VC) using
  the minimum number of vertices. In our problem which we call matching version of PVC, we have one more constraint where we require a certain lower bound for the size of the matching in the covered edges. We also consider the weighted variants of this problem. Our main goal is to investigate this problem in bipartite graphs. In this case, we refer to it as the matching version of the partial vertex cover problem in bipartite graphs (M-PVCB). Some applications of the PVCB and related problems are given in \cite{CGBS13,pvcbpaper}.

In the matching version of edge-weighted partial vertex cover problem (M-EPVC), we are given a
graph $G=(V,E)$, a weight function
$p:E\rightarrow \mathbb{N}$, and three positive integers $R$, $S$ and $T$. The goal is to
check whether there is a subset $V'\subseteq V$ of cardinality at most $R$,
such that the total weight of edges covered by $V'$ is at least $S$, and the covered edges include a matching of weight at least $T$.

In the present paper, we study the fixed-parameter tractability of M-EPVC in
bipartite graphs (M-EPVCB). By obtaining a W[1]-hardness result for so called multi-objective knapsack problem,
we show that M-EPVCB is W[1]-hard with respect to $R$. We also obtain similar results for other interesting parameters.
In case of some parameters we are able to show that M-EPVCB is Fixed-Parameter Tractable (FPT) with respect to the parameter under consideration. 

 The paper is organized as follows: Main notations and definitions are given in Section \ref{mnd}. The formal statements of the problems studied in the paper are given in Section \ref{sop}. Section \ref{motwork} presents the related approaches in the literature. In Section \ref{mkp}, we obtain some hardness results for so-called multi-objective knapsack problem and its two restrictions. These results are used later in Section \ref{mainres} where the main results of the paper are obtained. We conclude the paper in Section \ref{conc}, by summarizing our results and presenting some open problems that we feel deserve further investigation.
 
\section{Main Notations and Definitions}
\label{mnd}
 
We consider finite, undirected graphs that do not contain
loops or parallel edges. The degree of a vertex is the
number of edges of the graph that are incident to it. The maximum
degree of the graph $G$, denoted by $\Delta(G)$, is the maximum of all degrees of vertices
of $G$. For a positive integer $k$ we let $V_k$ and $V_{\geq k}$ be the sets of vertices of $G$ that have degree $k$ and at least $k$, respectively. Let $rad(G)$ and $diam(G)$ be the radius and diameter of $G$. If $P$ is a path of length $k$ in $G$, then we will say that $P$ is a $k$-path.

If $I$ is a subset of vertices of a graph $G$, then $I$ is called an independent set if any two vertices of $I$ are not adjacent in $G$. Let $\alpha(G)$ be the cardinality of the largest independent set of $G$. A subset $M$ of edges of $G$ is called a matching, if no two edges of $M$ are incident to the same vertex of $G$. A matching $M$ of $G$ is called an induced matching, if $G$ contains no path of length three, such that its first and third edge belong to $M$. Let $\nu(G)$ be the size of a largest matching of $G$, and let $\nu_{ind}(G)$ be the size of a largest induced matching of $G$. Clearly, in any graph $G$ we have $\nu_{ind}(G)\leq \nu(G)\leq \frac{|V|}{2}$. If $w:E(G)\rightarrow \mathbb{N}$ is a weight function defined on edges of $G$ and $X\subseteq E(G)$, then let $\nu_{w}(X)$ be the maximum weight (with respect to $w$) of a matching $M$, such that $M\subseteq X$. In particular, when $X=E(G)$, instead of writing $\nu_{w}(E(G))$, we will write $\nu_{w}(G)$.

A graph $G=(V,E)$ is bipartite, if its vertex set $V$ can be
partitioned into two independent sets $V_1$ and $V_2$. Usually, $V_1$ and $V_2$ are called the bipartition of $G$. If $G$ is a bipartite graph with a bipartition $V_1$ and $V_2$, such that any vertex of $V_1$ is adjacent to any vertex of $V_2$, then $G$ is called a complete bipartite graph. We will denote such a bipartite graph as $K_{|V_1|, |V_2|}$.

Given a graph $G= ( V, E)$, and a set ${ V_0 \subseteq V}$ of vertices,
an edge $(u,v) \in { E}$ is {covered} by ${ V_0}$ if $u \in { V_0}$
or $v \in { V_0}$. Let  $E(V_0)$ be the set of edges of $G$ that are
covered by $V_0$. The classical vertex cover
problem (VC) is defined as finding the smallest set $V_0$ of vertices of
the input graph $G=(V,E)$, such that $E(V_0)=E$. We will denote the cardinality of such set $V_0$ as $\tau(G)$. The vertex cover problem is a
well-known { NP-complete} problem \cite{Kar72}.

If $\Pi$ is an algorithmic problem and $t$ is a parameter, then the pair $(\Pi, t)$ is called a {parameterized} problem. The parameterized problem $(\Pi, t)$ is {fixed-parameter tractable} (or $\Pi$ is {fixed-parameter tractable with respect to the parameter} $t$) if there is an algorithm $A$ that solves $\Pi$ exactly, whose running-time is $g(t)\cdot poly(size)$. Here $g$ is some (computable) function of $t$, $size$ is the length of the input and $poly$ is a polynomial function. Usually, such an algorithm $A$ is called an FPT algorithm for $(\Pi, t)$. Sometimes we will say that $A$ runs in FPT($t$) time.

A parameterized problem is called {para}NP-{hard}, if it remains NP-hard even when the parameter under consideration is a constant. In the classical complexity theory, there is the notion of NP-hardness that indicates that a certain problem is less likely to be polynomial time solvable. It relies on the assumption $P\neq NP$. The classical {Satisfiability} problem is an NP-hard problem and any problem such that {Satisfiability} can be reduced to it is NP-hard, too. Similarly, in parameterized complexity theory there is the notion of W[1]-{hardness}, which indicates that a certain parameterized problem is less likely to be fixed-parameter tractable. It relies on the assumption $FPT\neq W[1]$, which says that not all problems from $W[1]$ are fixed-parameter tractable. The {Maximum Clique} problem where the parameter under consideration is $k$ - the size of the clique that we are looking for, is an example of a W[1]-hard problem, and any problem such that the maximum clique with respect to $k$ can be FPT-reduced to it, is also W[1]-hard. Recall that an {FPT reduction} between two parameterized problems $(\Pi_1, t_1)$ and $(\Pi_2, t_2)$ is an algorithm $R$ that maps instances of $\Pi_1$ to those of $\Pi_2$, such that 
\begin{enumerate}
    \item [(i)] for any instance $I_1\in \Pi_1$, we have $I_1$ is a ``yes"-instance of $\Pi_1$ if and only if $R(I_1)$ is a ``yes"-instance of $\Pi_2$,
    
    \item [(ii)] there is a computable function $h$, such that for any instance $I_1\in \Pi_1$ $t_2(R(I_1))\leq h(t_1(I_1))$,
    
    \item [(iii)] there is a computable function $g$, such that $R$ runs in time $g(t_1)\cdot poly(size)$.
\end{enumerate}

The reader can learn more about this topic from \cite{ParamBook15}, that can be a good guide for algorithmic concepts that are not defined in this paper.

\section{Formal Statement of Main Problems}
\label{sop}

In this paper, we study the following variants of the VC problem:
\begin{enumerate}[(a)]
\item The partial vertex cover problem (PVC)
\begin{definition}
Given a graph ${ G=( V,E)}$, two positive integers $k_1$, and $k_2$. The goal is to check whether there is a subset $V_0$ of $V$, such that $|V_0| \le k_1$ and $|E(V_0)|\geq k_2$.
\end{definition}

\item The weighted partial vertex cover problem (WPVC)

\begin{definition}
\label{def:WPVCB} 
Given a graph ${ G=( V,E)}$,  weight-functions $c: { V} \rightarrow { N}$ and ${p}: { E} \rightarrow { N}$, two positive integers $k_1$ and $k_2$. The goal is to check whether there is a subset $V_0$ of $V$, such that $\sum_{v \in { V_0}} c(v) \le k_1$ and 
 $\sum_{e\in { E(V_0)}} p(e) \geq k_2$?
\end{definition}

\item The partial vertex cover problem on bipartite graphs (PVCB) - This is the restriction of the partial vertex cover
problem (PVC) to bipartite graphs.

\item The weighted partial vertex cover problem on bipartite graphs (WPVCB) - This is the restriction of the weighted partial vertex cover problem (WPVC)
to bipartite graphs.

\item The VPVCB problem, a special case of the WPVCB problem, where all the edge weights are set to $1$,

\item The EPVCB problem, a special case of the WPVCB problem, where all the vertex weights are set to $1$.

\item The PVCB problem, a special case of the WPVCB problem, where all the vertex and edge weights are set to $1$.

\item The partial vertex cover problem with a matching constraint (M-PVCB) - This is a variant of the PVCB problem, in which we are given
a third parameter $k_3$ and the goal is to find a vertex subset of cardinality at most $k_1$, covering at least $k_2$ edges, such that the covered edges include a matching of size at least $k_3$.

\item The edge-weighted partial vertex cover problem with a matching constraint (M-EPVCB) - This is a variant of the EPVCB problem, in which we are given
a third parameter $k_3$ and the goal is to find a vertex subset of cardinality at most $k_1$, such that the covered edges have weight at least $k_2$ and they include a matching of weight at least $k_3$.
\end{enumerate}

The main contributions of this paper are the following:
\begin{enumerate}[a.]
\item { W[1]-hardness} of the bi-objective knapsack problem and its two restrictions with respect to the budget $B$.

\item Reduction of the parameterized problem (M-EPVCB, $k_1$) to instances in which $k_1<k_3<k_2< k_3\cdot \Delta(G)$.

\item { W[1]-hardness} of the M-EPVCB problem with respect to $k_1$.

\item NP-hardness of M-EPVCB in complete bipartite graphs $K_{t,t}$ and the paraNP-hardness of this problem with respect to some parameters.

\item Hardness of M-EPVCB with respect to $|V|-2\nu_{ind}(G)$ under the assumption FPT$\neq$W[1].

\item Hardness of M-EPVCB in paths and cycles under the assumption FPT$\neq$W[1].

\item Fixed parameter tractability of M-EPVCB with respect to $|V_{\geq 2}|$ and some other parameters.

\end{enumerate}

\section{Related Work}
\label{motwork}

PVC represents a natural theoretical
generalization of VC. It has some practical
applications. Flow-based risk-assessment models in computational
systems can be viewed as instances of PVC \cite{CGBS13}. In particular, PVC has applications to computer security when the input is a bipartite graph \cite{pvcbpaper}.

 VC is polynomial-time solvable in bipartite graphs. However, the
computational complexity of PVC in bipartite graphs remained open until it was recently
shown to be { NP-hard} by several authors \cite{Apoll,pvcbpaper,CS14,Joret}.

VC has been intensively studied from the perspective of approximation algorithms. There are many $2$-approximation algorithms for VC (see, for example,  \cite{Vazirani}). \cite{Kar09} provides an approximation algorithm for the VC problem which has a factor $(2 - \theta(\frac{1}{\sqrt{\log n}}))$. This is the best known result for now. The VC problem is shown to be { APX-complete} in \cite{PY91}. Moreover, it cannot be
approximated within a factor of $1.3606$ under the assumption ${ P \neq NP}$ \cite{DS05}. Recently, in \cite{KMS18}, this lower bound was improved to $(\sqrt{2}-\epsilon)$ for any $\epsilon>0$. If Khot's unique
games conjecture is true, then VC cannot be approximated within any constant factor smaller than $2$ \cite{KR08}. \cite{KPS11} provides a $(\frac{4}{3}+\epsilon)$-approximation algorithm for WPVC when the input graph is bipartite. Here $\epsilon >0$ is any constant.

All hardness results for the VC problem directly apply to the PVC
problem because the PVC problem extends the VC
problem. The PVC problem and the partial-cover variants
of related graph problems have been extensively studied
\cite{Bla03,BshB98,KLR08,M09,KMR07,KMRR06}. For example, there is an
$O(n \cdot \log n +m)$-time $2$-approximation algorithm for PVC based on the
primal-dual method \cite{M09}. Moreover, there is a combinatorial
$2$-approximation algorithm \cite{BFMR10}. Both of the two algorithms are
for a more general soft-capacitated version of PVC. There are several $2$-approximations resulting from other approaches
\cite{Bar01,BshB98,Hoch98,GKS04}. Finally, note that the WPVC
problem for trees is studied in \cite{pvct}. The paper provides an FPTAS for the problem. Additionally, the paper provides a polynomial time algorithm for the case of unweighted vertices (edges may have weights).

Another problem with a tight relationship to WPVC is the so-called budgeted
maximum coverage problem (BMC). In the BMC problem one tries to find a
min-cost subset of vertices, such that the profit of covered edges is
maximized. It can be easily shown that the two problems are equivalent
from the perspective of exact solvability. The BMC problem for sets
(not necessarily graphs) admits a $(1 - \frac{1}{e})$-approximation
algorithm as shown in \cite{BMCPsets}. However, special cases that beat this bound are
rare. The pipage rounding technique gives a
$\frac{3}{4}$-approximation algorithm for the BMC problem on graphs
\cite{AgeevSvirid}. This is improved to $\frac{4}{5}$ for bipartite
graphs in \cite{Apoll2}. In \cite{pvcbpaper,CS14}, an $\frac{8}{9}$-approximation algorithm for the problem is presented
when the input graph is bipartite and the vertices are unweighted (edges may have weights). The
result is based on the natural linear-programming formulation of the problem. The constant $\frac{8}{9}$ matches the integrality gap of the linear
program used in the formulation. Recently, in \cite{paschos19}, V. Paschos presented a polynomial time approximation scheme for
the edge-weighted maximum coverage problem on bipartite graphs.

Another problem with a close relation to the BMC and PVC problems is the profit cover problem (PC). Like in the BMC and PVC problems, the PC problem does not require a solution that covers all the vertices of a graph. However, instead of minimizing the number of vertices that cover a given number of edges or maximizing the number of edges covered by a fixed number of vertices, the goal in the PC problem is to maximize the profit. It is defined as the difference between the number of covered edges and the number of vertices in the cover. The PC problem has been considered in \cite{SRHH02}. There, it is shown that there exists a $O(p \cdot n + 1.151^p)$ algorithm for the PC problem. Here $p$ is the desired profit.
   
The $2$-PVCB problem studied in \cite{iwoca20} is closely related to the constrained minimum vertex cover problem on bipartite graphs (MIN-CVCB). In the MIN-CVCB problem, we are given two parameters $k_A$ and $k_B$, and the goal is to find a cover of the bipartite graph $G=(A,B,E)$ (with a bipartition $A$ and $B$) using at most $k_A$ vertices from $A$ and at most $k_B$ vertices from $B$. The MIN-CVCB problem is {NP-complete} and can be solved by a fixed parameter tractable algorithm that runs in time $O(1.26^{k_A+k_B}+(k_A+k_B) \cdot |G|)$ as demonstrated in \cite{CK03}. The $2$-PVCB problem is a generalization of the MIN-CVCB problem in which one does not need to cover all of the edges of the input graph $G$. In \cite{iwoca20}, it is shown that this generalization makes the problem no longer fixed parameter tractable in $k_A$ and $k_B$ under the assumption { FPT$\neq$W[1]}.
   
In this paper, we address our problems from the perspective of fixed-parameter tractability. From this point of view, PVC is in some sense
more difficult than VC. For example, PVC is
{ W[1]-hard} with respect to $R$ (that is, the number of vertices in the cover) \cite{ParamBook15}. On the other hand, VC is FPT
\cite{ParamBook15,Guo05}.
      
In \cite{AM11} the decision version of WPVCB is considered. There, the authors show that
this problem is FPT with respect to the vertex budget $k_1$, when the vertices and edges
of the bipartite graph are unweighted. In \cite{iwoca20}, by extending the result of Amini et al. \cite{AM11}, it is shown that
the decision version of WPVCB is FPT with respect to $k_1$, if the
vertices have cost one, while the edges may have arbitrary weights. On the other hand, the problem is { W[1]-hard} for arbitrary vertex weights, even when edges have profit one \cite{iwoca20}. \cite{iwoca20} proves that for
bounded-degree graphs WPVC is FPT with respect to $k_1$. Similar conclusion holds for WPVC with respect to $k_2$. Finally, \cite{iwoca20} shows that M-PVCB is FPT with respect to the budget $k_1$. Terms and concepts that we do not define in the paper can be found in \cite{ParamBook15}.

\section{The multi-objective knapsack problem}
\label{mkp}

In this section, we consider the multi-objective version of the classical knapsack problem. We present some hardness results for this version.

Recall that in the ordinary version of the problem, on the input we are given $n$ items $A=\{a_1,...,a_n\}$, a cost function $c:A\rightarrow N$, a profit function $p:A\rightarrow N$ and two constants $S$ and $T$. The goal is to check whether there is a subset $A_0$ of $A$, such that $c(A_0)\leq S$ and $p(A_0)\geq T$. This problem is NP-complete. It is natural to consider the following extension of this problem where on the input we have many cost functions and many profit functions. The goal in this new version is to check whether there is a subset of items whose cost with respect to any of the cost functions is at most some given bound, and its profit with respect to any of the profit functions is at least some other bound. This version of the problem is called the multi-objective knapsack problem.

In this paper, we will need only the case of this problem where on the input we have one cost function that is identically one and two profit functions. Let us formulate this version precisely:\\

{\bf Problem:} We are given a set of items $A=\{a_1,...,a_n\}$, a constant $B$ (that we will call a budget), two constants $P_1$ and $P_2$, and two profit functions $pr_1:A\rightarrow N$ and $pr_2:A\rightarrow N$. The goal is to check whether there is a subset $S\subseteq A$ with $|S|\leq B$, such that $pr_1(S)\geq P_1$ and $pr_2(S)\geq P_2$.\\

In this paper, we call this problem bi-objective knapsack problem or BKP for short. Below we are going to obtain some hardness results for BKP and its restrictions.

In the ``compendium of parameterized problems" (page 92 of \cite{MC06}), SubSet Sum problem is defined, which is the following: we are given a set of integers $X=\{x_1,...,x_n\}$, an integer $s$ and a positive integer $k$. The goal is to check whether $X$ has a subset $X'$ of cardinality $k$ such that the sum of numbers in $X'$ is exactly $s$. In \cite{MC06} it is stated that this problem is W[1]-hard with respect to $k$. The reference given there is \cite{DF95} where the authors proved the W[1]-hardness of Sized SubSet Sum (see page 123 of the paper) which is the same problem except that all numbers involved are positive integers. Clearly the hardness of Sized SubSet Sum implies the hardness of SubSet Sum as the latter is just an extension of the former.

Now, let us use this in order to obtain a reduction for the multi-objective knapsack problem where the weight functions can take negative values too. Assume that we have an instance of the SubSet Sum. Consider elements $A=\{a_1,...,a_n\}$. Let $pr_1(a_j)=x_j$ and $pr_2(a_j)=-x_j$. Define $B=k$, $P_1=s$ and $P_2=-s$. Then we have
that there is $X'$, $|X'|=k$ such that $\sum_{x\in X'}x=s$ if and only if $\sum_{x\in X'}x\geq s$ and $\sum_{x\in X'}(-x)\geq -s$. The latter is true if and only if there is $T\subseteq A$ with $|T|=B$ such that $pr_1(T)\geq P_1$ and $pr_2(T)\geq P_2$.

Thus, this version of the problem is W[1]-hard with respect to $B$. This proof has two drawbacks. First, we required that $|T|=B$. This is not a problem if the weights are non-negative, however, when they can be negative, then conditions $|T|=B$ and $|T|\leq B$ are not the same. Second, the weights can take negative values which we did not assume in our initial formulation of BKP.

We can fix the above two drawbacks as follows. Let us FPT-reduce the above mentioned variant of the knapsack problem with negative weights to the case when everything is positive. If this is achieved, as a side effect, we will solve also the issue over $|T|=B$ and $|T|\leq B$.

\begin{theorem}
\label{thm:BKPW1hardnessB} BKP is W[1]-hard with respect to $B$.
\end{theorem}

\begin{proof} Assume that $A=\{a_1,...,a_n\}$, $1\leq B\leq n$, $pr_1, pr_2: A\rightarrow Z$ and $P_1, P_2\in Z$ are given. Let us reduce this to the case when everything is positive. For any $x\in A$ and $i=1,2$ define:
\[pr'_i(x)=pr_i(x) +Q_i\]
\[P'_i=P_i+B\cdot Q_i,\]
where
\[Q_i=1+\sum_{x\in A}|pr_i(x)|.\]
Observe that $Q_i>0$ and $pr'_i(x)>0$. We can assume that $-Q_i<P_i<Q_i$, as if $P_i\geq Q_i$ then we have a trivial ``no"-instance, and if $P_i\leq -Q_i$ we have a trivial ``yes"-instance. Thus, $P'_i>0$.

Let us show that there is $T\subseteq A$, with $|T|=B$ such that $pr_i(T)\geq P_i$ if and only if $pr'_i(T)\geq P'_i$. By our definitions, we have
\[pr'_i(T)=pr_i(T)+|T|\cdot Q_i=pr_i(T)+B\cdot Q_i.\]
Hence $pr'_i(T)\geq P'_i=P_i+B\cdot Q_i$ if and only if $pr_i(T)\geq P_i$.

Thus, this is a reduction. Moreover, observe that it is a polynomial time reduction. Since the value of $B$ is unchanged, we have an FPT reduction. The proof is complete.
\end{proof}

Now, we are going to show that BKP remains W[1]-hard even if we have some additional restrictions on the profit functions. We will need these results in order to obtain some our main results in the next section.

\begin{theorem}
\label{thm:BKPrestrict1W1hardnessB} BKP remains W[1]-hard with respect to $B$ even if $pr_1(x)-pr_2(x)\leq pr_2(x)< pr_1(x)$ for any $x\in A$.
\end{theorem}

\begin{proof} We reduce BKP to itself when this additional constraint is satisfied. Assume that $A=\{a_1,...,a_n\}$, $1\leq B\leq n$, $pr_1, pr_2: A\rightarrow N$ and $P_1, P_2\in N$ are given. First let us show that we can assume that for any $x\in A$ we have $pr_2(x)< pr_1(x)$. For this purpose, define
\[Q=1+\max_{x\in A}\left\lceil\frac{pr_2(x)}{pr_1(x)}\right\rceil.\]
Let $pr'_1(x)=Q\cdot pr_1(x)$, $pr'_2(x)= pr_2(x)$, $P'_1=Q\cdot P_1$ and $P'_2=P_2$. Let the budget $B$ remain unchanged. We have that there is $T\subseteq A$ with $|T|=B$ such that $pr'_i(T)\geq P'_i$ if and only if $pr_i(T)\geq P_i$ (since everything is unchanged or multiplied by the same constant). Clearly, the new instance can be obtained in polynomial time. Let us show that the new instance satisfies $pr'_2(x)< pr'_1(x)$ for any $x\in A$. The inequality $pr'_2(x)< pr'_1(x)$ is equivalent to $pr_2(x)< Q\cdot pr_1(x)$ or $\frac{pr_2(x)}{pr_1(x)}< Q$. The latter we always have by the definition of $Q$.

Thus, in the very beginning we can assume that for any $x\in A$ we have $pr_2(x)< pr_1(x)$. Let us show that we can assume the other inequality as well. For a given instance define:
\[P_0=\max\{1, \max_{x\in A}(pr_1(x)-2pr_2(x))\}.\]
Observe that $P_0\geq 1$ by definition. For any $x\in A$ and $i=1,2$ define:
\[pr'_i(x)=pr_i(x)+P_0,\]
\[P'_i=P_i+B\cdot P_0.\]
As we have added the same number to every profit, we have for any $x\in A$ $pr'_2(x)< pr'_1(x)$. Since in the problem we were looking for $|T|=B$, we can prove similarly (see the proof of Theorem \ref{thm:BKPW1hardnessB}) that there is such $T\subseteq A$ with $|T|=B$, such that $pr_i(T)\geq P_i$ if and only if $pr'_i(T)\geq P'_i$. Thus we have a reduction. Moreover it is a polynomial time reduction and the value of the budget is not changed. Thus, all we are left is to show that for any $x\in A$ we have $pr'_1(x)\leq 2\cdot pr'_2(x)$. By definition, the last inequality is the same that
\[pr_1(x)+P_0\leq 2\cdot (pr_2(x)+P_0),\]
or equivalently,
\[pr_1(x)-2\cdot pr_2(x)\leq P_0.\]
However, the last one is always true because of the definition of $P_0$. The proof is complete.
\end{proof}

\begin{theorem}
\label{thm:BKPrestrict2W1hardnessB} BKP remains W[1]-hard with respect to $B$ even if $pr_1(x)-pr_2(x)\leq pr_2(x)< pr_1(x)$ for any $x\in A$ and $\sum_{i=1}^n[pr_1(a_i)-pr_2(a_i)]<\min_{x\in A} pr_2(x)$.
\end{theorem}

\begin{proof} By the previous theorem, we know that BKP with the first condition satisfied remains W[1]-hard with respect to $B$. Let us reduce these instances to ones when the second condition in the theorem is satisfied, too.

Start with an instance of BKP with the first condition and let
\[T=\sum_{i=1}^n[pr_1(a_i)-pr_2(a_i)].\]
Now, let us define new profit functions and new lower bounds for our parameters in the following way:
\[pr'_i(x)=pr_i(x)+(T+1),\]
and
\[P'_i=P_i+B(T+1).\]
As before, we can show that this established a reduction. Moreover,
\[0<pr'_1(x)-pr'_2(x)\leq pr'_2(x)< pr'_1(x),\]
All we are left is to show that our new constraint is true in the new instances as well. We have:
\[\sum_{i=1}^n[pr'_1(a_i)-pr'_2(a_i)]=\sum_{i=1}^n[pr_1(a_i)-pr_2(a_i)]< T+1<pr'_2(z)\]
for any item $z$. Thus our new condition is satisfied. The proof is complete.
\end{proof}

\section{Main results}
\label{mainres}

In this section, we obtain our main results. Some of our proofs rely on the hardness results established in the previous section. We start with the following observation that will allow us to obtain some restrictions for the values of $k_1$.

\begin{observation}
\label{obs:k1tau(G)} The instances of M-EPVCB, in which $k_1\geq \tau(G)=\nu(G)$ can be solved in polynomial time.
\end{observation}

\begin{proof}
Since $G$ is bipartite, we can find a smallest vertex cover in it in polynomial time. Because of our assumption, it contains at most $k_1$ vertices. Thus, in order to solve the instance, it suffices to check whether $w(E)\geq k_2$ and whether the weight of maximum weighted matching is at least $k_3$. Clearly, this can be done in polynomial time. The proof is complete.
\end{proof}

\begin{observation}
\label{obs:k2k3} The instances of M-EPVCB, in which $k_2\leq k_3$ can be solved in polynomial time.
\end{observation}

\begin{proof} Assume that we have an instance with $k_3>k_2$. We claim that $(G,w, k_1, k_2, k_3)$ is a ``yes"-instance, if and only if $(G,w, k_1, k_3, k_3)$ is a ``yes"-instance. Let us assume that $(G,w, k_1, k_2, k_3)$ is a ``yes"-instance. Then since the covered edges include a matching of $w$-weight at least $k_3$, we have that the covered edges are of weight at least $k_3$. Thus, $(G,w, k_1, k_3, k_3)$ is a ``yes"-instance. On the other hand, if $(G,w, k_1, k_3, k_3)$ is a ``yes"-instance. Then for any $k_2 < k_3$, we have that $(G,w, k_1, k_2, k_3)$ is a ``yes"-instance. Thus, the instances with $k_3\geq k_2$ can be reduced to those with $k_3=k_2$.

Theorem 4.6 of \cite{MaxWeightedkmatching} implies that one can find a maximum weighted $k_1$-matching (if it exists) in polynomial time. Thus, in order to solve the case $k_2=k_3$ of our problem, we just need to find a maximum weighted $k_1$-matching and check whether its weight is at least $k_2$. Hence, the instances with $k_2=k_3$ are polynomial time solvable. The proof is complete.
\end{proof}

The next proposition allows us to reduce the solution of some instances of M-EPVCB to instances of EPVCB. Since EPVCB is FPT with respect to $k_1$ \cite{iwoca20}, these instances can be solved in FPT($k_1$) time.

\begin{proposition}
\label{prop:k2/k3Delta(G)} Let $(G, w, k_1, k_2, k_3)$ be an instance of M-EPVCB with $\frac{k_2}{k_3}\geq \Delta(G)$. Then it is a ``yes"-instance if and only if $(G, w, k_1, k_2)$ is a ``yes"-instance of EPVCB.
\end{proposition}
 
\begin{proof} One direction is trivial. Let us assume that $(G, w, k_1, k_2)$ is a ``yes"-instance of EPVCB and $V'$ is the corresponding partial cover. Since $G$ is a bipartite graph, we have that the graph $G[E_{V'}]$ is a bipartite subgraph of $G$ with maximum degree at most $\Delta(G)$. Hence by K\"onig's theorem, it is $\Delta(G)$-edge-colorable. Thus, we can write
\[G[E_{V'}]=M_1\cup ... \cup M_{\Delta(G)}.\]
Here $M_1,...,M_{\Delta(G)}$ are matchings, that form the color classes of $G[E_{V'}]$. Then:
\[\max_{1\leq j\leq \Delta(G)}w(M_j)\geq \frac{w(E(G[E_{V'}]))}{\Delta(G)}\geq \frac{k_2}{\Delta(G)}\geq k_3.\]
Thus, $(G, w, k_1, k_2, k_3)$ is a ``yes"-instance of the edge-weighted version of the matching problem. The proof is complete.
\end{proof}

\begin{theorem}
\label{thm:k3atmostk1} The instances of M-EPVCB, in which $k_3\leq k_1$ can be solved in FPT($k_1$) time.
\end{theorem}

\begin{proof} We follow the proof of Theorem 5 from \cite{iwoca20}. Let EPVCB($A, B$) be the FPT($A$) algorithm that solves EPVCB (see Theorem 1 of \cite{iwoca20}), and let $R$ be the smallest integer for which $EPVCB(R, k_2)$ is feasible. We can assume that $R<k_1$. Let $H$ be the subgraph induced on these edges of weight at least $\geq k_2$. By the classical K\"{o}nig theorem we have $\nu(G)=\tau(G)$ for any bipartite graph $G$. By Observation \ref{obs:k1tau(G)}, we can assume that $k_1<\tau(G)$.

Observe that we can assume that $R<k_3\leq k_1<\tau(G)$. If $R\geq k_3$, then clearly $H$ can be covered with at most $k_1$ vertices, it has weight at least $k_2$, and it has a matching of size $R\geq k_3$ hence of weight at least $k_3$. Since $\tau(H)=R<k_3\leq \tau(G)$, we have that $E(H)\neq E(G)$. Thus, there is as an edge $e$ lying outside $H$. Add $e$ to $H$. If $\tau(H)$ has increased by adding $e$, define $R:=R+1$, otherwise let $R$ be the same. Repeat this process of adding edges outside $H$. Since $\tau(H)=R<k_3\leq \tau(G)$ and at each step $\tau$ can increase by at most one, at some point we will arrive into $H$ such that $R=\tau(H)=k_3\leq k_1$. Observe that $H$ can be covered with at most $k_1$ vertices, it has weight at least $k_2$ and it contains a matching of size $k_3$, hence of weight at least $k_3$. Thus, the problem is a ``yes"-instance.

The running time of the above algorithm is FPT in $k_1$. In order to see this, just observe that we have at most $k_1$ calls of $EPVCB(R, k_2)$ and the operation of adding the edges to $H$ and computing the size of the smallest vertex cover and the largest matching in the bipartite graph $H$ can be carried out in polynomial time. The proof is complete.
\end{proof}

The four statements proved above imply 

\begin{corollary}
When parameterizing M-EPVCB with respect to $k_1$, one can focus on instances in which $k_1<k_3<k_2< k_3\cdot \Delta(G)$.
\end{corollary}

In \cite{iwoca20}, it is shown that M-PVCB is FPT with respect to $k_1$. On the other hand, in the same paper it is proved that VPVCB and WPVCB are W[1]-hard with respect to $k_1$. Thus, their matching extensions are W[1]-hard with respect to $k_1$, too. It is interesting to wonder whether M-EPVCB is FPT with respect to $k_1$. Our next result addresses this question.

\begin{theorem} 
\label{thm:W1hardnessk1} M-EPVCB is W[1]-hard with respect to $k_1$.
\end{theorem}

\begin{proof} By Theorem \ref{thm:BKPrestrict1W1hardnessB}, BKP remains W[1]-hard even if for any item $x\in A$ we have 

\[pr_1(x)-pr_2(x)\leq pr_2(x)< pr_1(x).\]
Let us reduce this problem to M-EPVCB. For each item $x\in A$ put a 2-path (a path of length two) with edge-weights $w(e)=pr_1(x)-pr_2(x)$ and $w(e)=pr_2(x)$ for edges $e$ of the 2-path. Observe that we have an edge-weighted bipartite graph. Let $k_1=B$, $k_2=P_1$ and $k_3=P_2$. Let us show that we have a reduction. Observe that we can always avoid taking the degree-one vertices in our cover as we can simply take the degree-two vertex instead of it. Thus, we have a bijection among subsets of items and subsets of degree-two vertices of our bipartite graph. Moreover, for any subset $S\subseteq A$, we have
\[w(E_S)=\sum_{x\in S}[(pr_1(x)-pr_2(x))+pr_2(x)]=pr_1(S),\]
and
\[\nu_{w}(E_S)=\sum_{x\in S}\max\{pr_1(x)-pr_2(x), pr_2(x)\}=pr_2(S),\]
as we were considering the restriction of the multi-objective knapsack problem in which $pr_1(x)-pr_2(x)\leq pr_2(x)$ for any $x\in A$. Thus,
\[pr_1(S)\geq P_1 \textrm{ if and only if }w(E_S)\geq k_2,\]
and 
\[pr_2(S)\geq P_2 \textrm{ if and only if } \nu_{w}(E_S)\geq k_3.\]
Thus, we have a polynomial-time reduction. Observe that $k_1=B$, thus we have an FPT reduction. The proof is complete.
\end{proof}

In the reduction presented above, we have bounded maximum degree. Actually it is two. Thus, we have

\begin{corollary}
\label{cor:MaxDegreeFPT} Under the assumption FPT$\neq$W[1], M-EPVCB is not FPT with respect to $\Delta$.
\end{corollary}

Since $\Delta(G)=2$ in the above reduction we have that M-EPVCB is W[1]-hard with respect to $k_1+\Delta(G)$. On the other hand, because of Observation \ref{obs:k1tau(G)}, we can always assume that $k_1<\tau(G)=\nu(G)$. Thus, one may wonder whether M-EPVCB is FPT with respect to $\nu(G)+\Delta(G)$. Since by K\"onig's theorem, $E(G)$ can be partitioned into $\Delta(G)$ matching, we have that $|E|\leq \Delta(G)\cdot \nu(G)$. Thus, M-EPVCB is FPT with respect to the mentioned parameter.

Though M-EPVCB is W[1]-hard with respect to $k_1$, it is easy to show that it is FPT with respect to the complementary parameter $|V|-k_1$. In order to see this, in the given instance we can check whether $|V|-k_1\geq \frac{|V|}{2}$. If it holds, then $|V|\leq 2\cdot (|V|-k_1)$. In this case, we can generate all subsets $X$ of $V$, and for each $X$, we check that $|X|\leq k_1$, it has coverage at least $k_2$ and the includes a matching of weight at least $k_3$. These conditions can be checked in polynomial time. Thus, this case we can overcome in FPT($|V|-k_1$) time. On the other hand, if $|V|-k_1\leq \frac{|V|}{2}$, then $k_1\geq \frac{|V|}{2}$. Since in any bipartite graph $G$ $\tau(G)\leq \frac{|V|}{2}$ (just take the smallest set in the bipartition of $G$), we have $\tau(G)\leq k_1$. Observation \ref{obs:k1tau(G)} implies that these instances can be solved in polynomial time.

PVCB and its weighted extensions considered in this paper are NP-hard. Below we show that M-EPVCB remains hard for very restricted subclass of regular bipartite graphs.

\begin{theorem}
\label{thm:bipregular} M-EPVCB is NP-hard in bipartite regular graphs.
\end{theorem}

\begin{proof} We reduce M-EPVCB to its restriction in bipartite regular graphs. For a given weighted bipartite graph, define $C=|V|^2$. First, let us embed our bipartite graph $G$ into a bipartite $\Delta(G)$-regular graph $G'$ in a standard way. That is, we add new isolated vertices to the smallest set in the bipartition of $G$, so that two sets have equal size. Then we start adding new edges so that the graph remains bipartite and becomes regular. The number of newly added edges is less than $(|V|-1)^2<|V|^2=C$ as the size of larger part of $G$ is at most $|V|-1$ and minimum degree we can assume to be at least 1 and maximum degree is at most $|V|-1$. Next, we define the new weight function on $G'$ as follows: $w'(e)=C\cdot w(e)$ for old edges of $G$ and $w'(e)=1$ for new edges of $G$. Finally, for a given $k_1, k_2, k_3$ define $k'_1=k_1$, $k'_2=Ck_2$ and $k'_3=Ck_3$. Since $\log C\leq poly(size)$, we have that we have increased the parameters polynomially. Observe that the reduction is polynomial time. 

Let us show $(G, w, k_1, k_2, k_3)$ is a ``yes"-instance, if and only if $(G', w', k'_1, k'_2, k'_3)$ is a ``yes"-instance. If in $G$ we have a feasible solution then clearly it is feasible in $G'$ as everything is multiplied by $C$. Now, assume that $V_0$ is a feasible set in $G'$. Let $E_1$ be the set of new edges. We have
$w'(E_{V_0})\geq k'_2=Ck_2$ and $\nu_{w'}(E_{V_0})\geq k'_3=Ck_3$. Hence
\[w'(E_{V_0}\backslash E_1)\geq Ck_2-|E_1|>Ck_2-|V|^2=Ck_2-C=C(k_2-1)\]
and similarly
\[\nu_{w'}(E_{V_0}\backslash E_1)\geq Ck_3-|E_1|>Ck_3-|V|^2=Ck_3-C=C(k_3-1).\]
Thus,
\[w(E_{V_0}\cap E(G))=w(E_{V_0})>k_2-1\]
and 
\[\nu_{w}(E_{V_0}\cap E(G))=\nu_{w}(E_{V_0})>k_3-1.\]
Thus, $w(E_{V_0})\geq k_2$ and $\nu_{w}(E_{V_0})\geq k_3$. We have reduced the edge-weighted matching problem to its restriction in bipartite regular graphs in polynomial time. Thus, the problem is NP-hard in bipartite regular graphs. The proof is complete.
\end{proof}

\begin{remark}
Observe that the new vertices in the reduction do not play a role since they do not cover the old edges. Hence if we have a coverage greater than $C(k_2-1)$ in $G'$, then all these edges will be covered with old vertices. Thus they will give rise to a coverage larger than $k_2-1$ in $G$ with vertices of $G$.
\end{remark}

The strategy of the proof of the previous theorem implies the following corollary:

\begin{corollary}
\label{cor:completebipKtt} M-EPVCB remains NP-hard in complete bipartite graphs $K_{t,t}$.
\end{corollary} Just observe that if we want to obtain a complete bipartite graph in the reduction, we only need to continue adding edges of weight 1.

The proved hardness result has some consequences. In complete bipartite graphs $K_{t,t}$, we have that $\nu_{ind}(G)=1$, the domination number is 2, $2\alpha(G)-|V|=0$ ($\alpha(G)$ is the size of the largest independent set in $G$) and $\Delta(G)\cdot \nu(G)-|E(G)|=0$ (this is true for any bipartite regular graph). Thus, M-EPVCB is paraNP-hard with respect to these parameters. Below we obtain a hardness result with respect to $|V|-2\nu_{ind}(G)$.

\begin{theorem}
\label{thm:|V|minues2nuind(G)} Under the assumption FPT$\neq$W[1], M-EPVCB cannot be FPT with respect to $|V|-2\nu_{ind}(G)$.
\end{theorem}

\begin{proof} We reduce from the restriction of BKP from Theorem \ref{thm:BKPrestrict2W1hardnessB}. For a given instance of BKP from this theorem, consider the disjoint 2-paths from the proof of Theorem \ref{thm:W1hardnessk1}. Now, from each 2-path take one vertex of degree 1, and identify these $n$ vertices in order to get the tree $G'$. The resulting vertex $z$ in $G'$ has degree $n$. Observe that $|V(G')|-2\nu_{ind}(G')=1$.

Let us show that the instance of BKP is a ``yes"-instance, if and only if the instance of M-EPVCB is a ``yes"-instance. Forward direction is trivial. Let us prove the converse statement. Assume that the instance of M-EPVCB is a ``yes"-instance. It suffices to show that there is a feasible set that does not take $z$. Assume that we have a feasible set $V_0$ in M-EPVCB. We can assume that $z\in V_0$. If all neighbors of $z$ are in $V_0$, then we can remove it without losing feasibility. Thus, we can assume that at least one neighbor of $z$ does not belong to $V_0$. Replace $z$ with this neighbor in order to obtain a set $V_1$. Observe that because of the condition 
\[\sum_{i=1}^n[pr_1(a_i)-pr_2(a_i)]<\min_{x\in A} pr_2(x)\]
we have the same lower bounds for the coverage and the maximum weighted matching of covered edges. Thus, we have a feasible set that avoids $z$. The proof is complete.
\end{proof}

In the previous theorem we showed that M-EPVCB is hard already when $|V|-2\nu_{ind}(G)=1$. One may wonder what happens when $|V|-2\nu_{ind}(G)=0$. In this case, we have an induced perfect matching in $G$. Thus, $G$ is 1-regular. Therefore, by taking $k_1$ edges of maximum weight, we can check whether these edges have coverage at least $k_2$ and $k_3$. If they do, we have a ``yes"-instance. Otherwise, it is a ``no"-instance. Clearly, this can be done in polynomial time.

\begin{theorem}
\label{thm:pathscycles} Under the assumption FPT$\neq$W[1], M-EPVCB is not polynomial time solvable in paths and cycles.
\end{theorem}

\begin{proof} In Theorem \ref{thm:W1hardnessk1}, we have shown that M-EPVCB is W[1]-hard with respect to $k_1$ in vertex-disjoint 2-paths. Now, we are going to reduce these instances to cycles and paths in polynomial time. Clearly, this will prove the statement.

We follow the strategy of the proof of Theorem \ref{thm:bipregular}. Assume that $G$ is a vertex union of 2-paths. Let us take a constant $C=|V|^2$ and define the new values of parameters as we did in the proof of Theorem \ref{thm:bipregular}. Now, in order to obtain cycles or paths, we add edges of weight 1. Since $G$ is of maximum degree two, this is always possible. As in Theorem \ref{thm:bipregular}, one can prove that the original instance is a ``yes"-instance, if and only if the new instance is a ``yes"-instance. The proof is complete.
\end{proof}

In Theorem \ref{thm:|V|minues2nuind(G)}, we proved that M-EPVCB remains hard in a class of trees of radius two and diameter four. This implies that under the assumption FPT$\neq$W[1], M-EPVCB cannot be FPT with respect to $diam(G)$ and $rad(G)$. One may ask question about the parameters $|V|-diam(G)$ and $|V|-rad(G)$. Observe that in paths, we have that $|V|-diam(G)$ is constant. Thus, under the assumption FPT$\neq$W[1], M-EPVCB cannot be FPT with respect to it. On the other hand, for any graph $G$, we have $rad(G)\leq \frac{|V|}{2}$. Thus, $|V|-rad(G)\geq \frac{|V|}{2}$. Thus, M-EPVCB is FPT with respect to $|V|-rad(G)$. Finally, let us note that in paths $|V|-2\cdot rad(G)$ is constant, too. Thus, under the assumption FPT$\neq$W[1], M-EPVCB cannot be FPT with respect to it, too.

Observe that M-EPVCB is hard with respect to $|V_1|$ as cycles demonstrate. In these instances we have $|V_1|=0$. Also, observe that the problem is hard with respect to $|V_{\geq 3}|$ as paths demonstrate. 

\begin{theorem}
\label{thm:|V|geq2} M-EPVCB is FPT with respect to $|V_{\geq 2}|$.
\end{theorem}

\begin{proof} For a given instance of M-EPVCB, we consider two cases. If $\log |V|\leq |V_{\geq 2}|$, then $|V|$ is bounded in terms of our parameter. Thus, we can solve these instances in FPT($|V_{\geq 2}|$) time. Now assume that $\log |V|\geq |V_{\geq 2}|$. Observe that our graph can be represented as the vertices of $V_{\geq 2}$ that may or may not be joined to some vertices of $V_1$, plus we may have isolated edges (we can ignore isolated vertices). Observe that if our problem is a yes-instance, then there is a solution that takes some vertices from $V_{\geq 2}$ plus some independent vertices from the vertices of these isolated edges. Thus, we can consider the following simple algorithm: let us generate all subsets $X$ of $V_{\geq 2}$ that have size at most $k_1$. We have $2^{|V_{\geq 2}|}\leq |V|$ possibilities. For each of these choices we add $k_1-|X|$ independent vertices from isolated edges that have the largest coverage. We test the resulting set for feasibility. Clearly, this algorithm solves our problem exactly. Moreover, in this case the running time is polynomial. Thus, M-EPVCB is FPT with respect to $|V_{\geq 2}|$. The proof is complete.
\end{proof}

\section{Future Work}
\label{conc}

In this paper, we have shown that M-EPVCB is W[1]-hard or paraNP-hard with respect to many parameters. We also observed that in case of some parameters the problem is FPT. There are questions that deserve further investigation. Below we present some of them.

We have observed that M-EPVCB is FPT with respect to $|V|-k_1$. It would be interesting to investigate its hardness with respect to $|V|-2\cdot k_1$. 

In Theorem \ref{thm:W1hardnessk1}, we have shown that M-EPVCB is W[1]-hard with respect to $k_1$. It would be interesting to strengthen this result and show

\begin{conjecture}
\label{conj:minXY} M-EPVCB is W[1]-hard with respect to $\min\{|X|, |Y|\}$.
\end{conjecture}

In Corollary \ref{cor:MaxDegreeFPT}, we have shown that M-EPVCB is less likely to be FPT with respect to $\Delta(G)$. We suspect that
\begin{conjecture}
\label{conj:VminusDeltaFPT} M-EPVCB is W[1]-hard with respect to $|V|-\Delta(G)$.
\end{conjecture}

We can show that Conjecture \ref{conj:minXY} implies Conjecture \ref{conj:VminusDeltaFPT}. It suffices to present an FPT reduction from M-EPVCB considered with respect to $\min\{|X|, |Y|\}$ to M-EPVCB considered with respect to $|V|-\Delta(G)$. Assume that we have $G$ and let $(X, Y)$ be the bipartition of $G$. Assume that $|Y|\leq |X|$. Add a new vertex $z$ to $Y$ and add $k=|V(G)|$ vertices to $X$. Finally join $z$ to all vertices of $X$ both old and new. Let $H$ be the resulting bipartite graph. Observe that its size is polynomial in terms of $G$. Moreover, as we did in the proof of Theorem \ref{thm:bipregular}, we can multiply the weights of edges of $G$ with a big constant $C$ and define the new edges of $H$ to have weight one. As in the proof of this theorem, we define the parameters $k_1, k_2$ and $k_3$ in the same way. One can show that originally we had a ``yes"-instance, if and only if the new instance is a ``yes"-instance. Moreover, observe that this reduction is polynomial time. It remains to bound the parameters. Observe that
\[|V(H)|=|V(G)|+1+k\]
and
\[\Delta(H)=|X|+k.\]
Thus,
\[|V(H)|-\Delta(H)=|Y|+1=\min\{|X|, |Y|\}+1.\]
Thus, this reduction is an FPT-reduction. We finish the discussion with the following

\begin{observation}
\label{obs:Equivalence} The hardness of M-EPVCB with respect to $\min\{|X|, |Y|\}$ is equivalent to that of M-EPVCB with respect to $\nu(G)$.
\end{observation}

\begin{proof} Since $\nu(G)\leq \min\{|X|, |Y|\}$, one direction is trivial. Assume that the problem is hard with respect to $\min\{|X|, |Y|\}$. Let us show that it is hard with respect to $\nu(G)$. Let us embed $G=(X, Y, E)$ into a complete bipartite graph $K_{|X|, |Y|}$. As previously, we multiply edge-weights of old edges with a big, but polynomially bounded constant (see Theorem \ref{thm:bipregular}). The weights of new edges are one. As before, we have a polynomial time reduction. Thus, the two conditions of FPT-reductions are satisfied. It suffices to show that $\nu(K_{|X|, |Y|})$ is bounded in terms of a function of $\min\{|X|, |Y|\}$. We have
\[\nu(K_{|X|, |Y|})=\min\{|X|, |Y|\}.\]
Thus, the described reduction is an FPT reduction. The proof is complete.
\end{proof}

%\section*{References}

%% References


\bibliographystyle{elsarticle-num}

% \bibliographystyle{elsarticle-harv}
% \bibliographystyle{elsarticle-num-names}
% \bibliographystyle{model1a-num-names}
% \bibliographystyle{model1b-num-names}
% \bibliographystyle{model1c-num-names}
% \bibliographystyle{model2-names}
% \bibliographystyle{model3a-num-names}
% \bibliographystyle{model3-num-names}
% \bibliographystyle{model4-names}
% \bibliographystyle{model5-names}
% \bibliographystyle{model6-num-names}

%\bibliography{sample}

\end{document}